\documentclass[copyright,creativecommons,sharealike,noncommercial]{eptcs}
\providecommand{\event}{GandALF 2011} %
\usepackage{breakurl}             %

\usepackage{mathtools}
\usepackage{amsmath,amsfonts,amssymb,amsthm}

\usepackage{tikz}
\usetikzlibrary{shapes}
\usetikzlibrary{arrows}
\tikzstyle{max}=[shape=rectangle,draw,inner sep=0pt,minimum size=5mm,thick]
\tikzstyle{min}=[shape=diamond,draw,inner sep=0pt,minimum size=5mm,thick]
\tikzstyle{ran}=[shape=circle,draw,inner sep=0pt,minimum size=5mm,thick]

\newtheorem{theorem}{Theorem}
\newtheorem{lemma}{Lemma}
\newtheorem{fact}{Fact}

\newtheorem{corollary}{Corollary}
\newtheorem{conjecture}{Conjecture}
\newtheorem{proposition}{Proposition}
\newtheorem{assumption}{Assumption}

\theoremstyle{definition}
\newtheorem{definition}{Definition}
\newtheorem{example}{Example}

\newcommand{\coloneqq}{\mathrel{\mathop:}=}
\newcommand{\rst}[1]{\ensuremath{{\mathbin\upharpoonright}%
\raise-.5ex\hbox{$#1$}}}

\newcommand{\G}{\mathcal{G}}
\renewcommand{\H}{\mathcal{H}}
\newcommand{\V}{\mathcal{V}}
\newcommand{\eps}{\varepsilon}

\newcommand{\genGt}{\rhd}

\newcommand{\genNGt}{\ntriangleright}

\newcommand{\val}[1]{\mathit{Val}({#1})}  %
\newcommand{\run}[1]{\mathit{Run}(#1)}  %
\newcommand{\fpath}[2]{\mathit{FP}_{#1}(#2)} %
\newcommand{\Prb}[2]{\mathbb{P}_{#1}^{#2}}
\renewcommand{\Pr}[3]{\Prb{#1}{#2}\hspace{-0.16em}\left[{#3}\right]}   %
\newcommand{\Exp}{\mathbb{E}}
\newcommand{\Ex}[3]{\Exp_{#1}^{#2}\hspace{-0.16em}\left[{#3}\right]}   %
\newcommand{\len}[1]{\mathit{len}(#1)}

\newcommand{\LimInf}[2]{{(\mathit{LimInf}#1 #2)}}
\newcommand{\CN}{{\LimInf{=}{{-}\infty}}}

\newcommand{\genTran}[2]{%
    {}\mathchoice%
    {\stackrel{#1}{#2}}
    {\mathop {\smash{#2}}\limits^{\vrule width 0pt height 0pt depth 4pt\smash{#1}}}
    {\stackrel{#1}{#2}}
    {\stackrel{#1}{#2}}
{}}
\newcommand{\tran}[1]{\genTran{#1}{\rightarrow}}
\newcommand{\btran}[1]{\genTran{#1}{\hookrightarrow}}
\newcommand{\ctran}[1]{\genTran{#1}{\mapsto}}
\newcommand{\ftran}[1]{\genTran{#1}{\leadsto}}

\newcommand{\sP}[1]{#1_0}
\newcommand{\sMax}[1]{#1_1}
\newcommand{\sMin}[1]{#1_2}
\newcommand{\game}{(S,\tran{},\delta)}
\newcommand{\gameP}{(S',\btran{},\delta')} %

\title{
Optimal Strategies in
Infinite-state Stochastic
Reachability
Games
}
\def\titlerunning{
Optimal Strategies in
$\infty$-state Stochastic
Reachability
Games
}
\author{V\'aclav Bro\v{z}ek
\thanks{Author supported by the Newton International Fellowship of Royal Society.}
\institute{LFCS, School of Informatics\\ University of Edinburgh\\ Edinburgh, Scotland, UK}
\email{Vaclav.Brozek@ed.ac.uk}
}
\def\authorrunning{V. Bro\v{z}ek}
\begin{document}
\maketitle

\begin{abstract}
We consider perfect-information reachability stochastic
games for 2 players on infinite graphs.
We identify a subclass of such games,
and prove two interesting properties of it:
first, Player Max always has optimal
strategies in games from this subclass,
and second, these games are strongly
determined.
The subclass is defined by the property
that the set of all values can only have
one accumulation point -- 0.
Our results nicely mirror recent results
for finitely-branching games, where, on the contrary,
Player Min always has optimal strategies.
However, our proof methods are substantially different,
because the roles of the players are not symmetric.
We also do not restrict the branching of the games.
Finally, we apply our results in the context of recently
studied One-Counter stochastic games.

\end{abstract}

\section{Introduction}

Two-player turn-based zero-sum stochastic games,
simply called ``games'' in this text,
evolve randomly in discrete \emph{transitions} from one of countably many \emph{states}
to another. The winning condition is some property
of such infinite evolutions.
Each state is either owned by Player Max, Player Min,
or it is stochastic, and has a fixed set, possibly infinite,
of available outgoing transitions.
The states and transitions define a \emph{game graph},
an infinite path in this graph is called a \emph{run}.
The set of runs comes with a product topology
over the discrete state space, i.e.,
open sets are generated by sets of runs sharing a common finite prefix.
In stochastic states, the successor is sampled according
to a fixed distribution, whereas players choose successors
in states they own, based on the history of the play so far.
This induces a probabilistic measure for Borel-measurable
sets of runs in a natural way.

A winning condition is a set $W$ of runs.
A run from $W$ is won by Player Max,
the other runs are won by Player Min (the games are zero-sum).
For Borel measurable sets $W$, a fixed pair $(\sigma,\pi)$
of strategies for Player Max and Min, respectively,
and an initial state, $s$,
the probability that Max wins is denoted by
$\Pr{s}{\sigma,\pi}{W}$.
The \emph{value} of the game in $s$, denoted by $\val{s}$, is defined
as
\begin{equation}
\label{eq:det}
\val{s}
\coloneqq
\sup_\sigma
\inf_\pi
\Pr{s}{\sigma,\pi}{W}
=
\inf_\pi
\sup_\sigma
\Pr{s}{\sigma,\pi}{W}
.
\end{equation}
The above equality, a consequence of a more general, Blackwell-determinacy
result of Martin~\cite{M98}, implies that
for every $\eps>0$ both of the players have
so called \emph{$\eps$-optimal} strategies, $\sigma_\eps$
and $\pi_\eps$, such that
\(
\inf_\pi
\Pr{s}{\sigma_\eps,\pi}{W}
\geq
\val{s}-\eps
,
\)
and
\(
\sup_\sigma
\Pr{s}{\sigma,\pi_\eps}{W}
\leq
\val{s}+\eps
.
\)
This may not be true for the case when $\eps=0$, where
the optimal (i.e., $0$-optimal) strategies
may not exist for neither of the players.

We consider a stronger notion of determinacy than (\ref{eq:det}),
and call a game \emph{strongly determined} if for every state $s$,
every $\nu,\ 0\leq\nu\leq 1$,
and $\genGt\in\{>,\geq\}$
either Player Max has a strategy $\bar\sigma$
such that
\(
\forall \pi:
\Pr{s}{\bar\sigma,\pi}{W}
\genGt
\nu
,
\)
or Player Min has a strategy $\bar\pi$
such that
\(
\forall \sigma:
\Pr{s}{\sigma,\bar\pi}{W}
\genNGt
\nu
.
\)
Denote
\(
L\coloneqq
\sup_\sigma
\inf_\pi
\Pr{s}{\sigma,\pi}{W}
\)
and
\(
R\coloneqq
\inf_\pi
\sup_\sigma
\Pr{s}{\sigma,\pi}{W}
\),
then
if Max has a strategy $\bar\sigma$
such that
\(
\forall \pi:
\Pr{s}{\bar\sigma,\pi}{W}
\geq
\nu
\)
then $\nu\leq L$.
Similarly,
if Player Min has a strategy $\bar\pi$
such that
\(
\forall \sigma:
\Pr{s}{\sigma,\bar\pi}{W}
\leq
\nu
\)
then $\nu\geq R$.
By strong determinacy, $\forall \nu: \neg(R > \nu > L)$,
thus $R\leq L$.
$L \leq R$ follows from definitions, thus strong determinacy implies determinacy.
On the other hand, it is easy to see that the 
existence of $\eps$-optimal strategies
for both players implies strong determinacy
for cases where $|\nu-\val{s}|\geq2\eps$,
the players simply use their $\eps$-optimal strategies to win.
This works even for $\eps=0$, thus whenever both players
have optimal strategies, the game is strongly determined (for all $\nu$).
To sum up the relation between the key three notions:
Every game with a Borel winning condition is determined in the
sense of (\ref{eq:det}), some of these games
are strongly determined, and some of the strongly determined
games are those admitting optimal strategies for both players.
Example~\ref{ex:max-not-opt} and~\cite[Fig.~1]{BBKO11}
show that both the inclusions are proper.
More precisely, in the game from Example~\ref{ex:max-not-opt}, which we show later,
Player Min has only one (trivial) strategy, thus the game is strongly determined.
However, there is a state $r_0$, such that for every fixed
strategy of Max the probability of winning is strictly below $\val{r_0}$.
The game from~\cite[Fig.~1]{BBKO11}, is composed of two halves,
one of which is essentially equivalent to the game in Example~\ref{ex:max-not-opt},
and the other is a similar game adopted for Min (infinite branching needed).
As a consequence, neither Player Max in the first half, nor Min in the second half
have optimal strategies. Thus, fixing a strategy of one player first, which is $\eps$-optimal,
the other player may choose an $\eps/2$-optimal strategy to beat the first player.
As a consequence, no player has a winning strategy.

We are especially interested in the situation when $W$ is an
open set, and call such games \emph{open} as well.
This includes all
\emph{reachability} conditions, where $W$ is
the set of all runs visiting a state from a distinguished
set of target states, $T$.
For reachability, results of~\cite{BBKO11,BBKO09} imply (see Corollary~\ref{cor:min-opt})
that Player Min has always optimal strategies if every state, $s$, owned by Min
has at least one successor, $t$, such that $\val{s}=\val{t}$.
This is always the case in finitely-branching games,
where all states have only finite number of successors.
On the other hand, even in very simple reachability games
where every state has at most $2$ successors,
Player Max may not have an optimal strategy (cf.\ Example~\ref{ex:max-not-opt}).
Our main result gives a condition sufficient for the existence
of optimal strategies for Player Max.
\begin{theorem}
\label{thm:opt}
Let $\G$ be an open stochastic game.
Player Max has an optimal strategy in all
states,
if
\begin{equation}
\label{eq:acc-zero}
\tag{$*$}
\text{
the set
$V_\eps\coloneqq\{\val{s} \mid \text{$s$ is a state of $\G$}\land\val{s}\geq\eps\}$
is finite
for every $\eps>0$.}
\end{equation}
\end{theorem}
In particular, $\G$ is not assumed
to be finitely-branching.
Condition (\ref{eq:acc-zero}) is just saying that the set
$V\coloneqq\{\val{s} \mid \text{$s$ is a state}\}$
has no accumulation points, or the only such point is $0$.
It is a trivial task to construct a game
where none of the players owns a single state, i.e., a Markov chain,
and where the set $V$ contains other accumulation points than $0$.
In Markov chains, however, each player has only one, trivial, strategy,
which must thus be the optimal one.
This shows that (\ref{eq:acc-zero})
is not necessary.
However, there are at least two reasons for which (\ref{eq:acc-zero})
is interesting:
First, we identify a class of recently studied infinite-state
stochastic games which satisfy the assumption of Theorem~\ref{thm:opt},
and for which the existence of optimal strategies for Max was not
known before. This class, properly described later, consists of games
generated by One-Counter automata~\cite{BBEK11,BBE10,BBEKW10}, which satisfy a
certain additional property, which can be tested algorithmically.
As a special case, this class involves a maximizing variant of
Solvency Games~\cite{BKSV08}.

Second, in Examples~\ref{ex:max-not-opt} and~\ref{ex:noopt}, we show
games where Player Max lacks optimal strategies. These games
are rather simple, and violate (\ref{eq:acc-zero}) only ``very slightly'',
in particular, they
(1) are finitely-branching, and in fact have both the out-degree
and in-degree of the game graph bounded by $2$,
(2) do not contain states of Player Min at all,
(3) all transition probabilities in stochastic
states are uniformly distributed,
and
(4) $V$ has only one accumulation point.
This point is $1$ in Example~\ref{ex:max-not-opt},
and $1/2$ in Example~\ref{ex:noopt}. In the latter case,
the accumulation point is approached only from above, and
$V\cap[0,1/2)=\{0\}$.
Thus it is not possible to weaken the assumption (\ref{eq:acc-zero}) in Theorem~\ref{thm:opt}
by allowing other accumulation points than $0$.

As noted before, \emph{both} players having optimal strategies implies
strong determinacy.
But even for finitely-branching reachability games strong determinacy still holds, although
Player Max may not have optimal strategies, and only
Player Min always does~\cite{BBKO11}.
Interestingly,
we show here that under (\ref{eq:acc-zero}), where Max has optimal strategies,
and Min may not have such, strong determinacy survives.

\begin{theorem}
\label{thm:sdet}
Let $\G$ be an open stochastic game satisfying (\ref{eq:acc-zero}).
Then $\G$ is strongly determined.
\end{theorem}

\paragraph{Related work and open questions.}
Blackwell games are more general than our stochastic games, players there
choose their moves simultaneously, not knowing the concurrent choice of the opponent.
A famous determinacy result in the sense of (\ref{eq:det}) for Blackwell
games is given in~\cite{M98}.
Finitely-branching reachability games have been studied
as a theoretical background for some algorithmic results
concerning BPA games (i.e., games with graphs generated by stateless pushdown
automata) in~\cite{BBKO11, BBKO09}.
Finite-state reachability stochastic games
were studied in~\cite{C92}. In view of existence of optimal
strategies and strong determinacy, finite-state
games are not interesting: optimal strategies always exist there.
However, the precise complexity of associated computational problems for these games
is a long-standing and interesting open problem.

Theorem~\ref{thm:sdet} and the results
from~\cite{BBKO11,BBKO09} give us two classes of strongly determined games:
games satisfying (\ref{eq:acc-zero}), and finitely-branching games, respectively.
Neither of these two classes is contained in the other. The most interesting
question in our opinion is whether the following conjecture is true; and if it is not, for
which, as weak as possible, restrictions on $W$ and/or $\G$ it becomes true.
\begin{conjecture}
\label{conj:sdet}
Let $\G$ be a stochastic game, and $W$ a winning condition,
such that Player Max (or Player Min) has an optimal strategy
in every state of $\G$.
Then $\G$ is strongly determined.
\end{conjecture}
We do not even know whether the conjecture is true for all games
where $W$
is a reachability condition.
Other open questions include finding new interesting classes
of games where one of the players is guaranteed to
have optimal strategies.

\paragraph{Outline of the paper.}
We briefly formalise the necessary notions,
and recall some important known facts
in Section~\ref{sec:defs}.
In Section~\ref{sec:mdp} we prove Theorem~\ref{thm:opt} in the special
case of games without Player Min.
Both theorems are then proved in full generality
in Section~\ref{sec:games}.
Finally, in Section~\ref{sec:oc} we briefly explain
what are One Counter games, and apply our results to them.

\section{Preliminaries}
\label{sec:defs}

As noted in the Introduction, we use the simple term ``games''
for our special kind of games (Definition~\ref{def:game}).
Because we do not speak about other games here,
we hope the reader will excuse us for this inaccuracy.

\begin{definition}
A \emph{game graph}, $G=\game$, has a countable set
$S$ of states, partitioned into sets
$\sP{S}$, $\sMax{S}$, $\sMin{S}$ of stochastic states,
states of Player Max, and Player Min, respectively;
a countable \emph{transition relation} $\tran{}\subseteq S\times S$
such that $\forall r\in S:\exists s\in S:r\tran{}s$;
and a probability weight function $\delta:\sP{S}\times S \to [0,1]$ such that
for all $r\in\sP{S}$ we have
\(
\sum_{r\tran{}s}
\delta(r,s)
=
1
.
\)
\end{definition}

A \emph{run} is an infinite path in a game graph.
For a finite path $w$, we denote the states it visits by
$w(0),w(1),\ldots, w(k)$, and call $k=\len{w}$ the \emph{length} of $w$.
$\run{w}$ is the set of all runs extending $w$.
Unions of sets of the form $\run{w}$ are called open
sets, they are open in the product topology
over the discrete spaces $S$.
Closing the set of open sets under complements
and countable union defines the set of (Borel-)measurable
sets.

\begin{definition}
\label{def:game}
A \emph{game}, $\G$, is given by a game graph, $G$, and a Borel-measurable set
of runs, $W$, called the \emph{winning condition}.
If there is some $T\subseteq S$ so that
$W = \bigcup\{\run{w} \mid \text{$w$ ends in $T$}\}$
then $W$ is a \emph{reachability} condition, and
$\G$ is called a \emph{reachability game}.
\end{definition}

A strategy for Player Max is a function assigning to every
finite path (called a \emph{history}) ending in a state $s\in\sMax{S}$
a distribution over the successors of $s$. Similarly,
a strategy for Min is defined for histories ending in $\sMin{S}$.
A strategy is memoryless, if it only depends on the last state
of the history.

Fixing a pair of strategies, $(\sigma,\pi)$, for Max and Min,
respectively, we assign to every finite path, $w$, the product, $\rho^{\sigma,\pi}(w)$,
of weights on the edges along $w$ given by $\delta$,
$\sigma$, and $\pi$.
Fixing also an initial state, $s$, we define
a probability measure $\Pr{s}{\sigma,\pi}{\cdot}$ by
$\Pr{s}{\sigma,\pi}{\run{w}}\coloneqq 0$ for $w$ not starting in $s$,
$\Pr{s}{\sigma,\pi}{\run{w}}\coloneqq \rho^{\sigma,\pi}(w)$ for $w$ starting in $s$,
and extending this to complement and union
to satisfy the axioms of a probability measure.
The uniqueness of this construction is a standard fact, see, e.g.,
\cite[p.~30]{Puterman94}.

The definition of the value, $\val{\cdot}$, given in (\ref{eq:det}),
has thus been formalised.
For $\eps\geq0$, a strategy, $\sigma$, for Max is $\eps$-optimal in a state $s$
if
\(
\Pr{s}{\sigma,\pi}{W}
\geq
\val{s}-\eps
\)
for all strategies, $\pi$, for Min.
The $\eps$-optimal strategies for Min are defined analogously.
We call $0$-optimal strategies just optimal.

\subsection{Technical Assumptions}

Although a game graph, in general, may have an arbitrary structure,
we can always transform it to be a forest, without changing the properties
of the game, by keeping track of the history inside the states.
More precisely, 
given a game $\G=(G,W),\ G=\game$,
consider
a game $\G'=(G',W'),\ G'=\gameP$,
where the states in $S'$ are just finite sequences of states from $S$.
In particular, $S\subseteq S'$, and whenever $r\tran{} s$ in $\G$
then $w r \btran{} w r s$ in $\G'$.
Projecting the states of $S'$ to their last component induces a map, $\phi$,
from paths in $G'$ to paths in $G$.
We set $W'\coloneqq \phi^{-1}(W)$.
The map $\phi$ also induces a map, $\Phi$, from strategies in $\G$ to strategies
in $\G'$, by sending histories through $\phi$.
Naturally, the partition of $S'$, and the weight function $\delta'$ are both derived from $S$ and
$\delta$ by projecting states from $S'$ to the last component.

It is easy to verify that for every $s\in S$, if we restrict the game graphs
of $\G$ and $\G'$ to states reachable from $s$, then
$\phi$ is clearly bijective and preserves measurability in both directions.
Also $\Phi$ is bijective, and for all measurable $A\subseteq \run{s}$,
and all pairs $(\sigma,\pi)$ of strategies:
\(
\Pr{s}{\sigma,\pi}{A}
=
\Pr{s}{\Phi(\sigma),\Phi(\pi)}{\phi(A)}
.
\)
As a consequence,
$\val{s}$ is the same in $\G$ and $\G'$ for all $s\in S$,
and the sets of all values in $\G$ and in $\G'$ are equal.
Also, $W$ is open iff $W'$ is a reachability condition.
Every strategy in $\G'$ is memoryless, because $G'$ is a forest.
Finally, once we have a reachability objective, with the target
set $T$, we may clearly assume without loss of generality,
that all states in $T$ are absorbing. This shows that
to prove Theorems~\ref{thm:opt} and~\ref{thm:sdet} we may safely assume
the following:

\begin{assumption}
\label{as:tree}
The game graph is always a forest, all strategies are memoryless,
and the winning condition is a reachability condition specified
by some target set $T\subseteq S$, such that
for all $t\in T$ the only edge leaving $t$ is $t\tran{}t$.
\end{assumption}

\subsection{Known Results for Reachability Games}
We state here some known results to be used later.
The following gives a characterisation of values, and allows us
to characterise the existence of optimal strategies for Min.
\begin{fact}[\protect{cf.\ \cite[Theorem~3.1]{BBKO11}}]
\label{fa:fixp}
Let $\G=(G,W)$, $G=(S,\tran{},\delta)$ be a game,
with $W=\bigcup\{\run{w} \mid \text{$w$ ends in $T$}\}$.
The least fixed point of the following (Bellman) functional 
$\V:(S\to [0,1])\to (S\to [0,1])$ exists
and is equal to $\val{\cdot}$.
\[
\V(f)(s) =
\begin{cases}
1 & \text{if $s\in T$}  \\
\sup\{f(r) \mid s \tran{} r\} & \text{if $s\in\sMax{S}\setminus T$} \\
\inf\{f(r) \mid s \tran{} r\} & \text{if $s\in\sMin{S}\setminus T$} \\
\sum_{s \tran{}r} \delta(s,r) \cdot f(r) & \text{if $s\in\sP{S}\setminus T$}
\end{cases}
\]
\end{fact}

\begin{corollary}[\protect{cf.\ \cite[Theorem~3.1]{BBKO11}}]
\label{cor:min-opt}
Let $\G$ be a game as in Fact~\ref{fa:fixp}.
Let $G'=(S,\btran{},\delta)$ be a subgraph of $G$ where
$\btran{}$ is a subset of $\tran{}$, and if there is a pair $r,s\in S$
such that $r\tran{} s$ and $r \not\!\!\!\!\btran{}s$ then
$r\in\sMin{S}$ and there is some $s'\in S$ such that
$r\btran{}s'$ and $\val{s'}\leq\val{s}$ in $\G$.
Let $\G'=(G',W)$. Then the values are the same in $\G$ and $\G'$.

As a consequence,
a strategy, $\pi$, for Min is optimal iff for all
$r\in \sMin{S}$ it chooses with positive probability only successors $s\in S$ satisfying
$\val{r}=\val{s}$.
\end{corollary}
\begin{proof}
Let $\V'$ be the Bellman functional associated with $\G'$.
Observe that the values in $\G$ form a fixed point of $\V'$, thus
for all $s\in S$, $\val{s}$ in $\G'$ is equal to or less than $\val{s}$ in $\G'$.
Moreover, it cannot be less, because Player Max has the same set of strategies
in $\G'$ as in $\G$, whereas Player Min does not get more strategies in $\G'$.
To derive the consequence, remove all edges not used by $\pi$.
\end{proof}
Note that the situation is not symmetric for Player Max.
Consider games without Player Min, and with out-degree
and in-degree bounded by $2$. In particular, this
implies that every state, $r$, of Player Max has at least
one successor, $s$, with $\val{r}=\val{s}$.
Even in these games, Player Max may lack optimal strategies,
as illustrated in the following classical
(see, e.g., \cite[p.~871]{BBEKW10},\cite[Example~6]{BBFK08})
example.

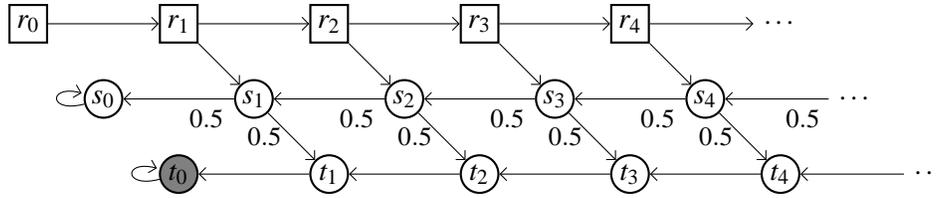
\begin{figure}
\begin{center}
\begin{tikzpicture}[x=2cm,y=2cm,>=angle 90]

\foreach \n/\m/\k in {0/0.5/1,1/1.5/2,2/2.5/3,3/3.5/4,4/4.5/5}
{
\node	(r\n)	at (\n, 4.5)	[max]	{$r_\n$};
\node	(s\n)	at (\m, 4)	[ran]	{$s_\n$};
}
\node	(t0)	at (1, 3.5)	[ran,fill=gray]	{$t_0$};
\node	(r5)	at (5, 4.5)		{$\cdots$};
\node	(s5)	at (5.5, 4)		{$\cdots$};
\node	(t5)	at (6, 3.5)		{$\cdots$};
\foreach \n/\m/\k in {1/1.5/2,2/2.5/3,3/3.5/4,4/4.5/5}
{
\node	(t\n)	at (\k, 3.5)	[ran]	{$t_\n$};
\draw[->] (r\n) to (s\n);
\draw[->] (s\n) to node [left] {{\small $0.5$}} (t\n);
}
\foreach \n/\nn in {0/1,1/2,2/3,3/4,4/5}
{
\draw[->] (r\n) to (r\nn);
\draw[->] (s\nn) to node [near start,below] {{\small $0.5$}} (s\n);
\draw[->] (t\nn) to (t\n);
}
\draw[->] (s0) to [loop left] (s0);
\draw[->] (t0) to [loop left] (t0);
\end{tikzpicture}
\end{center}
\caption{A reachability game where Player Max ($\Box$ states) has no optimal strategy.}
\label{fig:max-not-opt}
\end{figure}
\begin{example}
\label{ex:max-not-opt}
Consider the reachability game from Figure~\ref{fig:max-not-opt}.
Its game graph, $G$, has the set
\(
S
\coloneqq
\{
r_i, s_i, t_i
\mid
i \geq 0
\}
\)
of states,
partitioned by
$\sP{S}=\{s_i, t_i \mid i \geq 0\}$,
$\sMax{S}=\{r_i \mid i \geq 0\}$, and
$\sMin{S}=\emptyset$.
Transitions are
$s_0\tran{} s_0$, $t_0\tran{}t_0$,
and
$r_{i-1}\tran{}r_i$, $r_i\tran{}s_i$,
$s_i\tran{}s_{i-1}$,
$s_i\tran{}t_i$,
and $t_i\tran{}t_{i-1}$ for $i>0$.
Probabilities are always uniform.
The target set is $T=\{t_0\}$.
Clearly,
$\val{s_i}=1-2^{-i}$ for all $i\geq0$.
Thus $\val{r_i}=1$ for all $i\geq0$: for every $N> 0$,
choosing the transition $r_i\tran{}r_{i+1}$ for $i< N$,
and the transition $r_i\tran{}s_i$ for $i\geq N$,
is a $2^{-N}$-optimal strategy for Max.
Yet Max has no optimal strategy in any $r_i$, $i\geq 0$:
no strategy reaching some $s_j$ is optimal, and, on the other hand,
never reaching $s_j$ means never reaching $t$.
\end{example}

\section{Games without Player Min}
\label{sec:mdp}

\begin{proposition}
\label{prop:max-opt}
Let $\G=(G,W)$ be a stochastic game,
where
$G=(S,\tran{},\delta)$ and $\sMin{S}=\emptyset$,
\footnote{
These are also sometimes called
(minimizing) Markov Decision Processes (MDPs), see, e.g., \cite{BBEK11,BBE10,Puterman94}.
}
and $W$ is open.
If
(\ref{eq:acc-zero}) from Theorem~\ref{thm:opt}
is satisfied
then Player Max has an optimal strategy in all
states.
\end{proposition}

We fix the game $\G$ from Proposition~\ref{prop:max-opt}
in the rest of this section, devoted to proving the proposition.
By
Assumption~\ref{as:tree}, $G$ is a forest,
and there is $T\subseteq S$ such that
$W=\bigcup\{\run{w} \mid \text{$w$ ends in $T$}\}$
and for all $t\in T$ there is only one transition: $t\tran{} t$.
The proof is by contradiction, in three steps.
First, we prove that if there is a state with no optimal strategy,
then there must be a state from which winning with probability sufficiently close
to the optimum implies the need to use some value decreasing transition.
A transition $r\tran{}s$ is value decreasing if $\val{r}>\val{s}$.
Second, we will argue that the potential ``damage'' caused by this transition is
positive and bounded away from $0$, independently of the actual strategy.
Third, we show that (\ref{eq:acc-zero}) implies that the potential ``damage''
factor is indeed bounding the probability of reaching $T$ away from the
value, which is a contradiction with the definition of the value.

We introduce a random variable, $L$ (for ``loss'').
For a run, $\omega$, a \emph{losing} index is every $i$, such that
$\omega(i)\in \sMax{S}$ and
$\val{\omega(i)}>\val{\omega(i+1)}$.
If there is no losing index for $\omega$, we set $L(\omega)\coloneqq 0$.
Otherwise, there is the least losing index, $i$, and we set
\(
L(\omega)
\coloneqq
\val{\omega(i)} > 0.
\)
Finally, we say that a state $s \in S$ is \emph{losing}
if there is some $\delta_s>0$ such that
for every $\delta_s$-optimal strategy, $\sigma$, in $s$,
we have
\(
\Pr{s}\sigma{L>0} > 0.
\)

\begin{lemma}
\label{lem:non-opt}
Assume (\ref{eq:acc-zero}).
If $\exists s\in S$
such that
\(
\forall \sigma:
\Pr{s}\sigma{W} < \val{s}
\)
then there is also some losing state.
\end{lemma}

\begin{proof}
By contradiction.
Assume there is no losing state, we construct
an optimal strategy in every state.
Define a subset $\btran{}$ of the transition relation $\tran{}$ of $\G$,
by setting for every pair $r,s \in S$:
$r \btran{} s$ iff $r \tran{} s$ and
either $r\in \sP{S}$, or
$\val{r} = \val{s}$.
Observe that (\ref{eq:acc-zero}) implies
that for all $r\in \sMax{S}$
there is at least one $s$ such that
$r \btran{} s$ and $\val{r} = \val{s}$.
Thus $\btran{}$ is total and $G'=(S,\btran{},\delta)$ is a game graph.
Without losing states,
for every $r\in S$ and every $\eps>0$ there
is some $\eps$-optimal strategy, $\sigma$, such that
$\Pr{s}\sigma{L>0}=0$, i.e., $\sigma$
does not use value-decreasing transitions.
This strategy works in $\G'=(G',W)$ as well,
winning with the same probability, as in $\G$.
The values in $\G$ and $\G'$ are thus the same.

Consider now $\G'$.
Denote by $\fpath{k}{s}$ the set of all finite paths of length $k$ starting in $s$.
Due to the last sentence in Assumption~\ref{as:tree},
and because $\btran{}$ preserves value,
the following is true in $\G'$:
\begin{equation}
\label{eq:pres}
\forall k\geq 0:
\forall \sigma:
\forall s\in S:
\val{s}
=
\sum_{w\in \fpath{k}{s}}
\Pr{s}\sigma{\run{w}}
\cdot
\val{w(k)}.
\end{equation}
For all $s\in S$ fix a $1/4\cdot\val{s}$-optimal
strategy $\sigma_s$. After some
$n_s\geq0$ of steps, $T$ must be reached from $s$
under $\sigma_s$ with probability at least $\val{s}/2$, as
\(
\Pr{s}{\sigma_s}{W}
=
\lim_{k\to\infty}
\Pr{s}\sigma{
\{
\run{w}
\mid
\len{w}\leq k
\land
w(k)\in T
\}
}
.
\)

For all $s\in S$ we finally construct
a strategy $\sigma$ for $\G'$, optimal in $s$.
Because the values are the same in $\G$ and $\G'$, and every strategy for $\G'$
is also a strategy for $\G$, this will finish the proof of the lemma.
The strategy $\sigma$
starts in $s$ according to $\sigma_s$, and follows it for $n_s$ steps.
After that, having arrived to some state $r$,
it switches to $\sigma_r$ and follows it for other $n_r$ steps.
This is repeated ad infinitum.
The invariant (\ref{eq:pres}), and the choice of $n_r$ and $\sigma_r$
for $r\in S$, guarantee that
after the $m$-th stage of the above repetitive process, $T$ has actually been
reached with probability $(1-2^{-m})\cdot\val{s}$, proving that $\sigma$ is
optimal.
\end{proof}

For every losing state, $s \in S$, 
and every constant $\eps>0$
we define
\(
\ell_s^\eps
\coloneqq
\inf \{ \Ex{s}\sigma{L} \mid \text{$\sigma $ is $\eps$-optimal in $s$}\}.
\)
Since $\ell_s^\eps \leq \ell_s^\zeta\leq 1$ for $\eps \geq \zeta$,
the limit
\(
\ell_s \coloneqq
\lim_{\eps \to 0}
\ell_s^\eps
\)
exists.

\begin{lemma}
\label{lem:losing-bound}
Assume (\ref{eq:acc-zero}).
For every losing state, $s$, in $\G$ we have
\(
\ell_s
>0.
\)
\end{lemma}

\begin{proof}
By contradiction.
Assume that $s$ is losing and $\ell_s=0$.
To every strategy $\sigma$ which may possibly use
value-decreasing transitions $r\tran{} r'$
where $\val{r}>\val{r'}$ we consider a strategy
$\bar\sigma$, which copies the moves of $\sigma$ until a value-decreasing
transition is chosen. From that point on, just before the value-decreasing transition,
the strategy $\bar\sigma$ keeps choosing
arbitrary successors with the only requirement that they preserve the value,
i.e., whenever $\bar\sigma$ chooses a transition $s \tran{} s'$ with a positive
probability, $\val{s}=\val{s'}$.
Such a choice always exists, because $\sup_{s\tran{}s'}\val{s'}=\val{s}$, and
either $\val{s}=0$, in which case $\val{s'}=0$ for all $s',\ s\tran{}s'$, or
$\val{s}>0$, and by (\ref{eq:acc-zero}) $\val{s}>0$ cannot be an accumulation
point, so there is some $s'$, $s\tran{} s'$ with $\val{s}=\val{s'}$.
Observe that for every $\sigma$,
\(
\Pr{s}\sigma{W} - \Pr{s}{\bar\sigma}{W} \leq \Ex\sigma{s}{L}.
\)
As a consequence, due to $\ell_s=0$,
\(
\val{s}
=
\sup \{\Pr{s}{\bar\sigma}{W} \mid \sigma\text{ is some strategy}\}.
\)
This contradicts $s$ being losing, since
\(
\Pr{s}{\bar\sigma}{L>0}=0
\)
for every $\sigma$.
\end{proof}

\begin{proof}[Proof of Proposition~\ref{prop:max-opt}]
By contradiction.
Assume (\ref{eq:acc-zero}), and that there is some
$r\in S$ with no strategy optimal in $r$.
By Lemma~\ref{lem:non-opt}, there is a losing state, $s\in S$.
By Lemma~\ref{lem:losing-bound},
$\ell_s>0$.
Choose some $\eps>0$ such that $\ell_s^\eps \geq \ell_s / 2 > 0$.
Thus under every $\eps$-optimal strategy, $\sigma$, with some positive probability,
$p>0$, a state $r\in \sMax{S}$ with $\val{r} \geq \ell_s^\eps$
is visited, and some transition $r \btran{} r'$ with
$\val{r'} < \val{r}$ is taken.
Observe that (\ref{eq:acc-zero}) gives us the following
``value-gap'':
\[
\delta
\coloneqq
\inf \{ |\val{r}-\val{r'}| \mid r,r' \in S, \val{r}\neq\val{r'}, \val{r}\geq\ell_s^\eps\}
>0.
\]
This allows us to bound $p$ independently of $\sigma$, since
\(
\ell_s^\eps
\leq
\Ex\sigma{s}{L}
\leq p\cdot 1 + (1-p)(\ell_s^\eps-\delta)
\)
and hence
\[
p \geq \frac{\delta}{1+\delta-\ell_s^\eps} > 0.
\]
Thus for every strategy, $\sigma$, we have that
\(
\val{s} - \Pr{s}\sigma{W} \geq \min \{ \eps, \delta\cdot p\} > 0.
\)
This clearly contradicts the definition of $\val{s}$. The proof is finished.
\end{proof}

\newcommand{\lazRW}{Laz}
\newcommand{\lazRun}[1]{\lazRW(#1)}
\newcommand{\laz}[1]{laz(#1)}
\newcommand{\last}[1]{last(#1)}

\section{Reachability Games}
\label{sec:games}

In this section we prove Theorems~\ref{thm:opt} and~\ref{thm:sdet}.
Let us fix a game $\G=(G,W)$, where $G=(S,\tran{},\delta)$,
satisfying Assumption~\ref{as:tree}.
Also assume that $W$ is open, and thus 
there is $T\subseteq S$ such that
$W=\bigcup\{\run{w} \mid \text{$w$ ends in $T$}\}$.
We call a state $s$ \emph{safe} if
\(
\text{$\forall \sigma$ for Max}:
\text{$\exists \pi_\sigma$ for Min}:
\Pr{s}{\sigma,\pi_\sigma}{W}=0
.
\)
The following lemma states the strong determinacy restricted
to states with value $0$, and will be useful in proving
each of both theorems.

\begin{lemma}
\label{lem:sdet-zero}
If $\G$ satisfies (\ref{eq:acc-zero})
then
for every safe $s\in S$:
\(
\text{$\exists \pi$ for Min}:
\text{$\forall \sigma$ for Max}:
\Pr{s}{\sigma,\pi}{W}=0
.
\)
\end{lemma}

\begin{proof}
We cut off some choices for Min in the game graph $G$ of $\G$,
and obtain its sub-graph $G'$, so that all states reachable in $G'$ from $s$ have value $0$
in $\G'=(G',W)$.
In particular, no run can satisfy $W$.
Because the choices of Max remain unrestricted
in $G'$, this ensures that the probability of $W$ is $0$ in $\G$ as well.
Let us proceed in more detail.

Observe that every safe state has value $0$, so no safe state
is in $T$.
Also, observe that for every safe $r\in\sP{S}\cup\sMax{S}$
and $s\in S$, if $r\tran{} s$ then $s$ is safe.
Likewise, if $r\in\sMin{S}$ is safe, then there must be a safe $s$ such that
$r\tran{} s$.
Fix a safe $s$, and define $G'$ as the smallest sub-graph of $G$
containing $s$ and satisfying that if $r$ is in $G'$, then so is
every safe successor $r'$ of $r$ in $G$.
As shown above, $G'$ is a game graph, the probability
assignment $\delta$ from $G$ is valid in $G'$ as well,
and all states in $G'$ are safe.
Hence, no paths in $G'$
visit $T$, and the value of every state in $\G'$ is $0$.
Fix an arbitrary strategy $\pi$ for Min in $\G'=(G',W)$,
then $\Pr{s}{\sigma,\pi}{W}=0$ for all $\sigma$ of Max in $\G'$.
All transitions out of safe states of Max
were preserved in $G'$,
and $\pi$ is also a strategy in $\G$, so we have
$\Pr{s}{\sigma,\pi}{W}=0$ also
for every $\sigma$ of Max in $\G$.
\end{proof}

\subsection{Proof of Theorem~\ref{thm:opt}}

\begin{lemma}
\label{lem:min-max-win}
If $\G$ satisfies (\ref{eq:acc-zero}),
then for all $s\in S$ we have:
\(
\text{$\forall \pi$ for Min}:
\text{$\exists \sigma$ for Max}:
\Pr{s}{\sigma,\pi}{W}\geq\val{s}
.
\)
\end{lemma}
\begin{proof}
For every (memoryless, due to Assumption~\ref{as:tree})
strategy $\pi$ of Player Min, we denote by $\G_\pi$
the game where the choices of Player Min are resolved using $\pi$.
Formally, $\G_\pi=(G',W)$, where $G'=(S',\btran{},\delta')$,
and (1) $S'=S$ but comes with a different partition:
$\sP{S'}=\sP{S}\cup\sMin{S}$,
$\sMax{S'}=\sMax{S}$,
$\sMin{S'}=\emptyset$,
(2) the relation $\btran{}\subseteq \tran{}$ is
given by
$r\btran{} s$ iff $r\tran{} s$ and either $r\in \sP{S}\cup\sMax{S}$,
or $r\in \sMin{S}$ and $\pi(r)(s)>0$,
and (3) $\delta'=\delta\cup\pi$.
For every strategy $\sigma$ for Player Max, and every $s\in S$
the measure $\Pr{s}{\sigma,\pi}{\cdot}$ in $\G$
obviously coincides with $\Pr{s}{\sigma}{\cdot}$ in $\G_\pi$.
Thus we may apply Proposition~\ref{prop:max-opt} to all
$\G_\pi$ to derive the lemma.
\end{proof}

Consider now the following game $\H=(H,W)$, which is a slight modification
of $\G$. The set of states of $H=(S,\btran{},\delta_H)$ is $S$, the same as in $G$, and
with the same partition.
There is a transition $r\btran{}s$ iff exactly one of these three situations
occurs:
$\val{r}=0$ in $\G$, and $s=r$;
or $\val{r}>0$, $r\in \sP{S}$ and $r \tran{} s$;
or $\val{r}>0$, $r\notin \sP{S}$, $r \tran{} s$,
and $\val{r}=\val{s}$ in $\G$.
In other words, in $H$ we made all states with value $0$ absorbing, and
only left value preserving transitions for players.
Finally, $\delta_H$ is the only probability weight function which coincides
with $\delta$ on stochastic states with positive value.

\begin{lemma}
\label{lem:val-pres}
If $\G$ satisfies (\ref{eq:acc-zero}), then
$H$ is a game graph, and the values
are the same in $\G$ and $\H$.
\end{lemma}

\begin{proof}
We refine the modifications from above into three steps,
obtaining game graphs $H_0=G$, $H_1$, $H_2$, and $H_3=H$.
We will show for each $i\in\{1,2,3\}$ that $H_i$ is a game graph,
and that the values are the same in $\H_i=(H_i,W)$ as they are in $\G$.
All the graphs constructed have the same set of states, $S$, and the same partition,
as $G$, and the same weight function, $\delta_H$, as $H$.

$H_1=(S,\ctran{},\delta_H)$, and
$r\ctran{}s$ iff
$\val{r}=0$ in $\G$, and $s=r$,
or $\val{r}>0$ and $r \tran{} s$.
$H_1$ is clearly a game graph, because $\ctran{}$ is total.
The values did not change, because each absorbing loop outside of $T$ has
value $0$.
Moreover, every $r\in\sMin{S}$ has always a successor with the same value.
Indeed,
if $\val{r}=0$ then $r$ itself is its own successor in $G_1$;
if $\val{r}>0$ then $\inf_{r\ctran{}s}\val{s}=\val{r}$, and
by (\ref{eq:acc-zero}), since $\val{r}>0$ cannot be an accumulation
point, there is some $s$, $r\ctran{} s$ with $\val{r}=\val{s}$.
By Corollary~\ref{cor:min-opt}, Min has optimal strategies
in $\H_1$.

$H_2=(S,\ftran{},\delta_H)$, and
$r\ftran{}s$ iff $r \ctran{} s$ and
either
$\val{r}=0$ in $\G$,
or $r\notin \sMin{S}$,
or (if $\val{r}>0$ and $r\in \sMin{S}$)
$\val{r}=\val{s}$ in $\G$.
Because Min has always value-preserving transitions in $\H_1$,
$H_2$ is clearly a game graph, and
by Corollary~\ref{cor:min-opt}
all strategies of Min in $\H_2$ are optimal.
Fix one such $\pi$ for Min, and an arbitrary $s\in S$.
By Lemma~\ref{lem:min-max-win} there is a $\sigma$ for Max
in $\G$ (and thus also in $\H_2=(H_2,W)$) such that
$\Pr{s}{\sigma,\pi}{W}\geq\val{s}$.
Because $\pi$ is optimal, $\sigma$ cannot choose
value-decreasing transitions.
Thus, even when only using edges in $\btran{}$, i.e., from $H_3=H$,
we still obtain that
$\inf_\pi\sup_\sigma\Pr{s}{\sigma,\pi}{W}=\val{s}$.
Thus also the graph $H$ is a game graph, and
the values in $\H$ and $\G$ are the same.
\end{proof}

\begin{lemma}
\label{lem:max-opt}
If $\G$ satisfies (\ref{eq:acc-zero}),
then Player Max has an optimal strategy, $\sigma$, in $\H$.
\end{lemma}

\begin{proof}
We first describe $\sigma$, then we prove that it is optimal.
In every state, $s$, there is some
$1/2\cdot\val{s}$-optimal strategy, $\tau_s$,
for Max.
We call a history (i.e., a finite path), $w$, starting in some state $s$,
and ending in $r$, \emph{lazy},
if
$\val{r}>0$ and
\(
\inf_\pi
\Pr{s}{\tau_s,\pi}{W \mid \run{w}}
=0
.
\)
Observe that each history, $w$, can be uniquely split into
a sequence of sub-paths, divided by single states,
$w = s_0 w_0 s_1 w_1 s_2 \cdots s_k w_k$,
$k\geq 1$, $s_i \in S$, $w_i \in S^*$,
such that
for all $i<k$, $s_i w_i s_{i+1}$ is lazy,
and
for all $i\leq k$, $s_i w_i$ is not lazy.
We call $k$ the \emph{laziness index} of $w$,
written $\laz{w}$
and $s_k w_k$ the \emph{non-lazy} suffix of $w$.
We now define $\sigma$ for a history $w$ with a non-lazy
suffix $s_k w_k$ by
\(
\sigma(w)
\coloneqq
\tau_{s_k}(s_k w_k)
.
\)

Now we prove that $\sigma$ is optimal.
To do so, we need to extend the laziness index to runs.
For a run, $\omega$, we set
\(
\lazRun\omega
\coloneqq
\sup
\{
\laz{w}
\mid
\omega \in \run{w}
\}
\in
\mathbb{N}\cup\{\infty\}
.
\)
Thus we defined a random variable, $\lazRW$.
We prove the following claim, which clearly implies
the statement of the lemma:
\begin{equation}
\label{eq:s-opt}
\forall s \in S:
\text{$\forall \pi$ for Min}:
\forall k\geq 0:
\Pr{s}{\sigma,\pi}{W \land \lazRW \leq k}
\geq
\val{s}
\cdot
(1-2^{-k})
.
\end{equation}
By induction on $k$. Fix some $s\in S$, and a strategy, $\pi$, for Min.
Clearly, (\ref{eq:s-opt}) is true for $k=0$.
Also it is true when $\val{s}=0$.
Assume thus $\val{s}>0$ and $k=\ell+1$ for some $\ell\geq0$.
We set $L$ to be the set of all finite paths, $w$,
such that $\laz{w}=k$ and the non-lazy suffix
only consists of one state.
Denote by $\last{w}$ the last state of $w$.
Observe that,
by the definition of $\sigma$ and $\tau_{s_k}$,
\begin{equation}
\label{eq:lazy-half}
\forall w\in L:
\text{$\forall \pi$ for Min}:
\Pr{s}{\sigma,\pi}{W \mid \run{w}}
\geq
1/2\cdot\val{\last{w}}
.
\end{equation}
Let $\Lambda$ be any prefix-free set of finite paths
such that $\Pr{s}{\sigma,\pi}{\bigcup_{w\in\Lambda}\run{w}}=1$.
Because $\H$ only contains value-preserving edges for players,
we have
\begin{equation}
\label{eq:pres2}
\val{s}
=
\sum_{w\in \Lambda}
\Pr{s}{\sigma,\pi}{\run{w}}
\cdot
\val{\last{w}}.
\end{equation}
We have
\(
p
\coloneqq
\Pr{s}{\sigma,\pi}{W \land \lazRW \leq \ell}
\geq
\val{s}\cdot(1 - 2^{-\ell})
,
\)
by the inductive hypothesis.
We also have
\(
q
\coloneqq
\sum_{w\in L}
\Pr{s}{\sigma,\pi}{\run{w}}
\cdot\val{\last{w}}
=\val{s}-p
,
\)
by (\ref{eq:pres2}).
By (\ref{eq:lazy-half}), $\Pr{s}{\sigma,\pi}{W \land \lazRW = k}\geq q\cdot 1/2$.
Finally,
\begin{align*}
\Pr{s}{\sigma,\pi}{W \land \lazRW \leq k}
&=
\Pr{s}{\sigma,\pi}{W \land \lazRW \leq \ell}
+
\Pr{s}{\sigma,\pi}{W \land \lazRW = k}
\\
&
=
p
+
q\cdot1/2
=
p+
(\val{s}-p)\cdot1/2
=
p/2 + \val{s}/2
\\
&
\geq
(2^{-1} - 2^{-(\ell+1)})\cdot\val{s}
+
\val{s}\cdot2^{-1}
=
(1 - 2^{-(\ell+1)})\cdot\val{s}
.
\end{align*}
\end{proof}

\begin{proof}[Proof of Theorem~\ref{thm:opt}]
Consider the strategy $\sigma$ from Lemma~\ref{lem:max-opt}.
It partially defines a strategy in $\G$.
To complete its definition, we now specify
it for histories containing a transition of the form $r\tran{}s$, where
$r\in\sMin{S}$ and $\val{s}>\val{r}$, by requiring
$\sigma$ to behave as a $1/2\cdot(\val{s}-\val{r})$-optimal
strategy since that point.
Fix an initial state, $s$, and consider an arbitrary strategy, $\pi$, of Min.
If $\pi$ is optimal, then it is also valid in $\H$,
and $\Pr{s}{\sigma,\pi}{W}=\val{s}$ by Lemmata~\ref{lem:val-pres}
and~\ref{lem:max-opt}.
For a non-optimal $\pi$ it is easy to verify that
$\Pr{s}{\sigma,\pi}{W}>\val{s}$ by both the definition of $\sigma$,
and Lemmata~\ref{lem:val-pres} and~\ref{lem:max-opt}.
\end{proof}

\subsection{Proof of Theorem~\ref{thm:sdet}}

If both players have optimal strategies, the game is strongly determined.
However, even under Condition~(\ref{eq:acc-zero}),
Player Min may not always have an optimal
strategy,
because of states with value $0$, without value-preserving
transition for Min available.
See the game in~\cite[Fig.~1]{BBKO11} restricted to states reachable from $s$,
for an example.
Theorem~\ref{thm:sdet} is a direct consequence of Lemma~\ref{lem:easy-det}
and Lemma~\ref{lem:hard-det},
where the former lemma deals with all ``easy cases'', and the latter
``patches'' the above deficiency
by using Lemma~\ref{lem:sdet-zero} to deal with states with value $0$,
and ``restoring'' the optimal strategies for both players in the rest.

\begin{lemma}
\label{lem:easy-det}
Assume that $\G$ satisfies (\ref{eq:acc-zero}).
Let $s\in S$, $0\leq\nu\leq 1$,
and $\genGt\in\{>,\geq\}$.
Assume that either $\val{s}=0$, or $\nu\neq\val{s}$,
or $\genGt=\geq$.
Then
either Player Max has a strategy $\bar\sigma$
such that
\(
\forall \pi:
\Pr{s}{\bar\sigma,\pi}{W}
\genGt
\nu
,
\)
or Player Min has a strategy $\bar\pi$
such that
\(
\forall \sigma:
\Pr{s}{\sigma,\bar\pi}{W}
\genNGt
\nu
.
\)
\end{lemma}
\begin{proof}
The case when $\nu=\val{s}=0$ is solved by Lemma~\ref{lem:sdet-zero}.
If $\nu<\val{s}$, we can choose any $1/2\cdot(\val{s}-\nu)$-optimal
strategy for Max as $\bar\sigma$.
Similarly,
if $\nu>\val{s}$, we can choose any $1/2\cdot(\nu-\val{s})$-optimal
strategy for Min as $\bar\pi$.
If $\nu=\val{s}$ and $\genGt=\geq$,
we can choose any optimal
strategy for Max as $\bar\sigma$.
Such a strategy exists due to Theorem~\ref{thm:opt}.
\end{proof}

It remains to solve $\nu=\val{s}>0$ and $\genGt=>$.
We do two preprocessing steps on $\G$ to first obtain $\G'$, and then $\H$.
In $\H$ both players will have optimal strategies,
and we will be able to lift such a strategy for Min back to $\G$
iff Max does not have a strategy ing $\G$ to always win with probability $>\val{s}$.

We fix $s\in S$ with $\val{s}>0$ and set
\(
R
\coloneqq
\{
r \in S
\mid
\val{r}=0
\land
\exists \bar\sigma:
\forall \pi:
\Pr{r}{\bar\sigma,\pi}{W} > 0
\}
.
\)
Intuitively,
if Max does not have a strategy to always win with probability $>\val{s}$,
then Min can always respond to a strategy
$\sigma$ of Max with a $\pi_\sigma$, so that $R$ is not visited
at all from $s$ under these strategies, and yet Max wins
with probability at most $\val{s}$.
Thus, if we cut off all states from $R$, producing the game $\H$,
we obtain a valid game graph, and the values of states will not change.

Before we describe this formally, we observe that neither of the players
benefits from using transitions which do not preserve the value.
Let $\G'=(G',W)$, $G'=(S,\ctran{},\delta)$ be a game given by restricting
the edges of $G$ to value-preserving where possible:
for all $r,r'\in S$ we require that
$r\ctran{} r'$ iff $r\tran{}r'$ and either $r\in \sP{S}\cup R$,
or $\val{r}=\val{r'}$.

\begin{lemma}
\label{lem:Gprime}
Assume that $\G$ satisfies (\ref{eq:acc-zero}).
Then the values in $\G$ and in $\G'$ are the same, and
for all $s\in S$, each of the following is true in $\G'$ if it is true in $\G$:
\begin{gather}
\label{eq:max-non-win}
\text{$\forall \sigma$ for Max}:
\text{$\exists \pi_\sigma$ for Min}:
\Pr{s}{\sigma,\pi_\sigma}{W}\leq\val{s}
,
\\
\label{eq:min-non-win}
\text{$\forall \pi$ for Min}:
\text{$\exists \sigma_\pi$ for Max}:
\Pr{s}{\sigma_\pi,\pi}{W}>\val{s}
.
\end{gather}
\end{lemma}

\begin{proof}
By Theorem~\ref{thm:opt}, there is an optimal strategy, $\sigma$, for Min.
This is also a strategy for $\G'$, thus for all $s\in S$, $\val{s}$
in $\G'$ is at least $\val{s}$ in $\G$.
On the other hand, by Corollary~\ref{cor:min-opt}, cutting off non-optimal
edges leaving states from $\sMin{S}\setminus R$ does not alter the values.
Further, cutting off non-optimal edges from $\sMax{S}$ could only decrease the values.
Thus, for all $s\in S$, the values in $\G'$ and $\G$ are equal.

Now we fix some $s\in S$, and prove that if (\ref{eq:max-non-win}) is true in $\G$
then it is true in $\G'$.
Let $\sigma$ be a strategy for Max in $\G'$, i.e., it is a strategy for $\G$ which does
not use value-decreasing edges.
If $\sigma$ is optimal, then the strategy $\pi_\sigma$ from (\ref{eq:max-non-win}) in $\G$
necessarily has to use value-preserving edges everywhere, and thus it is valid in $\G'$ as well.
If $\sigma$ is not optimal, consider again the response $\pi_\sigma$ of Min
to satisfy (\ref{eq:max-non-win}) in $\G$. If $\pi_\sigma$ cannot be used directly
in $\G'$, then there must be some $r\in \sMin{S}\setminus R$ where
$\pi_\sigma$ chooses a successor $r'$ with $\val{r'}>\val{r}$.
But because $r\notin R$, there must also be a successor $r''$
such that $\val{r''}=\val{r}$. We modify $\pi_\sigma$ to a $\pi'_\sigma$, which
chooses for all such $r$ the value-preserving successor instead of $r'$,
and continues as a $1/2\cdot(\val{r'}-\val{r})$-optimal strategy in $\G'$.
Clearly,
\(
\Pr{s}{\sigma,\pi'_\sigma}{W}
\leq
\Pr{s}{\sigma,\pi_\sigma}{W}
\)
in $\G$, and since $\pi'_\sigma$ is also a strategy in $\G'$, (\ref{eq:max-non-win}) is true in $\G'$
as well.

Finally, we prove that if (\ref{eq:min-non-win}) is true in $\G$
then it is true in $\G'$.
Let $\pi$ be a strategy in $\G'$.
Fix the choices of $\pi$ in $\G'$ outside of $R$ to define a game $\G_\pi$.
By Corollary~\ref{cor:min-opt}, $\G_\pi$ has the same values as $\G$.
Thus, optimal strategies of Max in $\G_\pi$ exist, because $\G_\pi$ satisfies (\ref{eq:acc-zero}),
and only choose edges preserving the value in $\G$.
Consider the strategy $\sigma_\pi$ witnessing (\ref{eq:max-non-win}) in $\G$.
We now define a strategy $\sigma'_\pi$ in $\G'$:
it copies moves of $\sigma_\pi$ in $\G'$, unless
$\sigma_\pi$ chooses some value-decreasing edge.
In that case, instead of following $\sigma_\pi$,
$\sigma'_\pi$ immediately switches to some optimal strategy
for $\G_\pi$.
Since the values in $\G_\pi$ are the same as in $\G$,
this only increases the probability of winning,
thus
\(
\Pr{s}{\sigma'_\pi,\pi}{W}
\geq
\Pr{s}{\sigma_\pi,\pi}{W}
.
\)
\end{proof}

\begin{lemma}
\label{lem:hard-det}
Assume that $\G$ satisfies (\ref{eq:acc-zero}).
For all $s\in S$ such that
$\val{s}>0$, if
\begin{equation}
\label{eq:max-min-lose}
\text{$\forall \sigma$ for Max}:
\text{$\exists \pi_\sigma$ for Min}:
\Pr{s}{\sigma,\pi_\sigma}{W}\leq\val{s}
,
\end{equation}
then
\begin{equation}
\label{eq:min-max-lose}
\text{$\exists \pi$ for Min}:
\text{$\forall \sigma$ for Max}:
\Pr{s}{\sigma,\pi}{W}\leq\val{s}
.
\end{equation}
\end{lemma}

\begin{proof}
By Lemma~\ref{lem:Gprime}, if
(\ref{eq:max-min-lose}) $\implies$ (\ref{eq:min-max-lose})
in $\G'$ then the implication holds in $\G$ as well,
and if $\G$ satisfies (\ref{eq:acc-zero}) then so does $\G'$.
Thus we focus on $\G'$ instead.
We describe the modification of $\G'$, called $\H$, where we cut off $R$.
By $Att(R)$ we denote the set of all states, $r$, such that in $\G'$
Max has a strategy, $\sigma$, such that for all $\pi$ for Min,
$\Pr{r}{\sigma,\pi}{\text{Reach $R$}}>0$.
Further, we consider the edge relation $\btran{}$, which is simply
the relation $\ctran{}$ without edges leading to
states from $Att(R)$.

We fix some $s$, $\val{s}>0$, satisfying (\ref{eq:max-min-lose}), and
by $S'$ we denote the subset of all $r\in S$
to which there is a path from $s$ in the graph $(S,\btran{})$.
Consider a game graph, $H=(S',\btran{},\delta)$, inheriting
the partition of states from $G$.
The edge relation is the $\btran{}$ defined above, only
restricted to $S'\times S'$.
Observe that if $r\in \sP{S'}\cup\sMax{S'}$ and $r\tran{} r'$ for some $r'\in S$,
then $r'\in S'$. This is because $r\notin Att(R)$ implies $r'\notin Att(R)$
if $r$ is not owned by Min.
Similarly, for all $r\in\sMin{S'}$ there is a $r'\in S'$
such that $r\btran{}r'$.
Thus $\delta$, restricted to $S'$, is still a valid probability
weight function, and $H$ is a valid game graph.
We abuse the letter $W$ to denote a restriction of $W$ to $\H$, and
define a game $\H=(H,W)$.

Because all edges leaving states from $\sMin{S}\setminus R$
were value-preserving in $\G'$, Corollary~\ref{cor:min-opt} yields that the values
stay the same in $\H$ as they were in $\G'$,
and there is an optimal strategy, $\bar\pi$ for Min in $\H$. This is also a strategy for $\G'$,
and because the choices of Player Max were not affected when reducing
$\G'$ to $\H$, we obtain, that for all $\sigma$ for Max we have
$\Pr{s}{\sigma,\bar\pi}{W}\leq\val{s}$ both in $\H$ and in $\G'$.
This proves (\ref{eq:min-max-lose}).
\end{proof}

\section{One Counter Games}
\label{sec:oc}

One Counter stochastic games (OC-SSGs),
see, e.g.,~\cite{BBEK11,BBE10,BBEKW10},
are games played on transition graphs of one-counter automata.
Such automata have a finite \emph{control-state} unit, $Q$,
and a set of rules, which are triples
of the form $(r,k,s)$ with $r,s\in Q$ and $k\in\{-1,0,+1\}$.
States of an OC-SSG are then of the form $s_n$ where
$s\in Q$ is a \emph{control state},
and $n\geq 0$ is an integer, representing the \emph{counter value}.
Transitions are generated
by setting $r_i\tran{}s_j$ if $i>0$ and there is a rule
$(r,j-i,s)$. Moreover, states with counter $0$ are
made absorbing, $s_0\tran{}s_0$, to reflect that the system halts
with the empty counter.
The partition of states is induced by a partition of $Q$,
and the probabilities of
transitions out of stochastic states,
are induced by probabilities on rules.
OC-SSGs come with an implicit reachability objective,
the set to be reached is the set $\{s_0 \mid s\in Q\}$
of states with counter $0$.
Because the system halts in $0$
we also call this a \emph{termination} winning condition.

\begin{figure}
\begin{minipage}[b]{0.5\linewidth}
\centering
\begin{tikzpicture}[x=1.5cm,y=1.5cm,>=angle 90]
\node	(s0)	at ( 3, 0)	[max,fill=gray]	{$s_0$};
\node	(d0)	at ( 1, 0)	[ran,fill=gray]	{$d_0$};
\node	(u0)	at ( 2, 0)	[ran,fill=gray]	{$u_0$};
\node	(r0)	at ( 4, 0)	[ran,fill=gray]	{$r_0$};
\node	(z0)	at ( 5, 0)	[ran,fill=gray]	{$z_0$};
\node	(t0)	at ( 6, 0)	[ran,fill=gray]	{$t_0$};

\foreach \n in {1,2,3}
{
\node	(s\n)	at ( 3, \n)	[max]	{$s_\n$};
\node	(d\n)	at ( 1, \n)	[ran]	{$d_\n$};
\node	(u\n)	at ( 2, \n)	[ran]	{$u_\n$};
\node	(r\n)	at ( 4, \n)	[ran]	{$r_\n$};
\node	(z\n)	at ( 5, \n)	[ran]	{$z_\n$};
\node	(t\n)	at ( 6, \n)	[ran]	{$t_\n$};
}

\node	(s4)	at ( 3, 4)		{$\vdots$};
\node	(d4)	at ( 1, 4)		{$\vdots$};
\node	(z4)	at ( 4, 4)		{$\vdots$};
\node	(t4)	at ( 6, 4)		{$\vdots$};

\foreach \x/\y/\z in {0/1/2,1/2/3,2/3/4}
{
\draw[->] (s\y) to (u\y);
\draw[->] (u\y) to (s\z);
\draw[->] (u\y) [bend right] to (d\y);
\draw[->] (d\y) [bend right] to (u\y);
\draw[->] (d\y) [out=270] to (s\x);
\draw[->] (s\y) to (r\y);
\draw[->] (r\y) to (z\y);
\draw[->] (r\y) [bend right] to (t\y);
\draw[->] (z\y) to (z\x);
\draw[->] (t\y) to (t\z);
}

\draw[->] (d4) [out=270] to (s3);
\draw[->] (z4) to (z3);

\draw[->] (s0) [loop right] to (s0);
\draw[->] (u0) [loop right] to (u0);
\draw[->] (d0) [loop right] to (d0);
\draw[->] (r0) [loop right] to (r0);
\draw[->] (z0) [loop right] to (z0);
\draw[->] (t0) [loop right] to (t0);
\end{tikzpicture}
\end{minipage}
\hspace{0.5cm}
\begin{minipage}[b]{0.4\linewidth}
\centering
\begin{tikzpicture}[x=1.7cm,y=1.7cm]
\node	(s)	at (0.5, 0)	[max]	{$s$};
\node	(rg)	at (  0, 1)	[ran]	{$d$};
\node	(rb)	at (  1, 1)	[ran]	{$u$};
\node	(t)	at (2.5, 0)	[ran]	{$r$};
\node	(g)	at (  2, 1)	[ran]	{$z$};
\node	(b)	at (  3, 1)	[ran]	{$t$};

\path[->] (s) edge [bend left=25] (rb)
 (rb) edge [bend left=30] node  [right] {$+1$} (s)
 (rg) edge [bend right=30] node [left] {$-1$} (s)
 (rb) edge [bend right=30] (rg)
 (rg) edge [bend right=30] (rb)
 (s) edge (t)
 (t) edge (g)
 (t) edge (b)
 (g) edge [loop above] node [above] {$-1$} (g)
 (b) edge [loop above] node [above] {$+1$} (b);
\end{tikzpicture}
\caption{Left: A game, $\G$, where player Max ($\Box$) does not have optimal strategies. All stochastic ($\bigcirc$) states
have uniform distribution on outgoing transitions.
Right: A One Counter description of $\G$.
Signed numbers represent counter increments.}
\label{fig:ocmdp-non-max}
\end{minipage}
\end{figure}
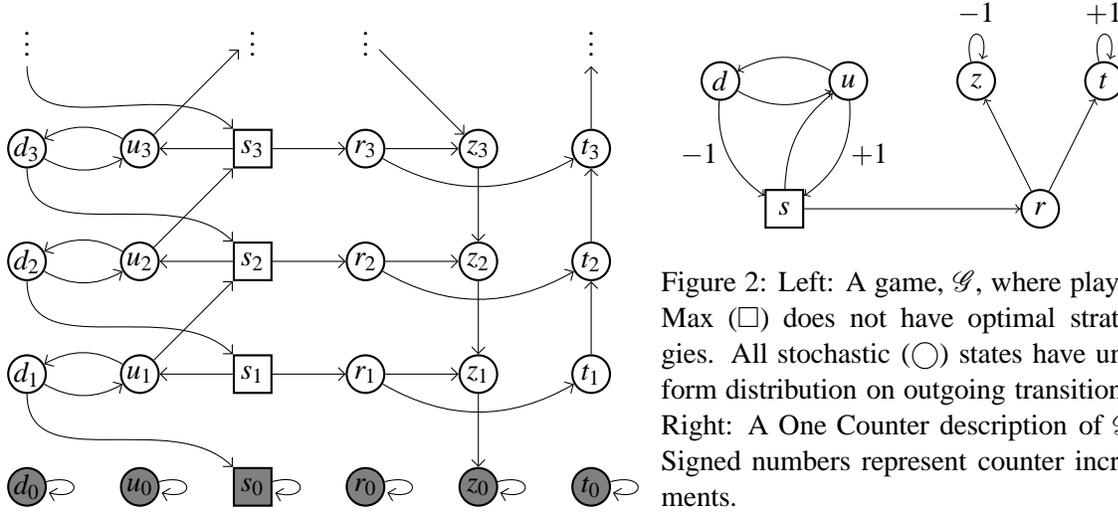

\begin{example}
\label{ex:noopt}
In the right-hand part of Figure~\ref{fig:ocmdp-non-max} we give
the one-counter automaton with the set $Q=\{s,u,d,r,z,t\}$ of control states.
An unlabelled edge, like $s\tran{} u$,
represents a $0$-rule, e.g., $(s,0,u)$.
A label ($\pm 1$) represents the counter change, e.g., the loop
$t\tran{}t$ represents $(t,+1,t)$. The square-state $s$ belongs
to Max, other states are stochastic. The distributions on outgoing transitions
are implicitly uniform in this example.
In the left-hand part is the generated OC-SSG.
Grey states are to be reached.
Later in this section we will show that $\val{s_i}=\frac{2^i+1}{2^{i+1}}$,
but no strategy of Player Max is optimal in $s_i$.
Observe that $1/2=\lim_{i\to\infty}\frac{2^i+1}{2^{i+1}}$ is an
accumulation point in the set of all values.
\end{example}

Note that every OC-SSG has bounded out-degree and in-degree, in particular
it is finitely branching. Thus Min has always optimal strategies
in OC-SSGs. However, they may not always satisfy
(\ref{eq:acc-zero}), and Example~\ref{ex:noopt} shows that
in OC-SSGs, Max may have no optimal strategies.
On the other hand, the structure of the accumulation points in
the set of all values is well understood for OC-SSGs.
To describe it, we need to introduce another winning objective.

In OC-SSGs there is an implicit boundary on the counter value -- if
it reaches zero, the system halts.
However, we may also interpret the one-counter automaton
as a directed graph on $Q$, with the rules as edges with \emph{rewards}.
This way we obtain a finite game graph.
Accumulating those rewards along a run in such a game graph
then corresponds to observing the counter in the OC-SSG, with the exception
that the counter does not stop in $0$ and may get negative.
Adding the winning condition (for Max) that the $\liminf$ of the
accumulated rewards be $-\infty$, we just defined $\CN$-games.

In~\cite{BBE10,BBEKW10} it was shown that both players always have pure and memoryless
optimal strategies in $\CN$-games, and the optimal value is always rational and computable.
Observe that the termination values, $\val{s_n}$, for a fixed $s\in Q$, are non-increasing with
increasing $n$. Thus their limit exists, and, in fact, it is an easy exercise to
employ the results of~\cite{BBE10,BBEKW10} to prove that the $\CN$-value of a control state, $s$,
equals $\lim_{n\to\infty}\val{s_n}$.
Intuitively this is because, with increasing the initial counter, $n$, the objective
of reaching $0$ becomes more and more similar to the $\CN$ objective.
Thus the set of $\CN$-values of all states $s\in Q$ contains
the set of all accumulation points of the termination values.
It is also possible to decide in time polynomial in $|Q|$
whether a $\CN$-value, $\nu$, actually is an accumulation point,
i.e., whether for all states, $s$, with $\CN$-value $\nu$
the limit of termination values stabilises after finitely many steps.

\begin{corollary}
\label{cor:oc}
Let $\G$ be an OC-SSG with the set $Q$ of control states.
If for every $s\in Q$ the $\CN$-value of $s$ is $1$ or $0$,
then Player Max has an optimal strategy for termination in $\G$.
\end{corollary}

\begin{proof}
The limits of termination values are approached from above,
because $\val{s_n}\geq\val{s_{n+1}}$ for all $s\in Q$ and all $n\geq 0$.
Thus, $1$ is not an accumulation point, and
we may apply Theorem~\ref{thm:opt}.
\end{proof}

Note that the class of OC-SSGs satisfying the condition of Corollary~\ref{cor:oc}
involves all OC-SSGs where
the graph of rules is strongly connected,
and one of the players is missing.
This is because $\CN$ is a prefix independent objective,
and the strong connectivity allows the only player to reach
each control state almost surely, thus all control
states have the same $\CN$-value. By results of~\cite[Theorem~3.2]{GH-SODA10},
such a common value can only be $0$ or $1$.
In particular, Corollary~\ref{cor:oc} covers both
the Solvency games, see~\cite{BKSV08}, and their maximizing variant.

In Solvency games, a gambler has an initial positive amount of money,
and in each step chooses one of finitely many actions. Each
action is associated with a distribution on a finite set of
integers. A number from this set is then sampled, and added
to the sum of money owned by the gambler (it can be, however, negative),
and the process ends only when the wealth becomes $\leq0$.
This is easily modelled by one-player OC-SSGs (see~\cite{BBEKW10}),
and these have strongly connected graphs of rules, because
the only state where the gambler chooses the action, is reachable
from all other states.
The natural scenario is, obviously, with Player Min for these games,
and there the existence of optimal strategies follows
from the finite branching. However, the dual situation, with Player
Max, is theoretically interesting as well, and we are not
aware of any result prior to our Corollary~\ref{cor:oc}, indicating the existence of optimal
strategies for Player Max.

\subsection{Analysis of Example~\ref{ex:noopt}}
Consider an arbitrary $n\geq 1$.
It is easy to see that $\val{r_n}=\frac{1}{2}$.
Observe that starting in $u_n$,
$s_{n+1}$ is visited with probability
$\sum_{i=0}^\infty2^{-1-2i}=\frac{2}{3}$, and $s_{n-1}$
with probability $\frac{1}{3}$.

\begin{lemma}
\label{lem:prob-oc}
For the unique strategy, $\sigma$, not using
transitions $s_n\tran{} r_n,\ n\geq1$, we have
\(
\Pr{s_i}\sigma{W}=2^{-i}.
\)
\end{lemma}
\begin{proof}
Clearly
\(
\Pr{s_0}\sigma{W}=1=2^{-0}.
\)
Further, the assignment
\(
x
\coloneqq
\Pr{s_1}\sigma{W}
\)
is the least non-negative solution
of the equation $x = \frac{2}{3} + \frac{x^2}{3}$,
see, e.g.,~\cite[Theorem~3.4]{EKM04} or \cite[Theorem~1]{EY05stacs},
which is $\frac{1}{2}$.
Solving the recurrence
\(
\Pr{s_i}\sigma{W}
=
\frac{2}{3}
\cdot
\Pr{s_{i-1}}\sigma{W}
+
\frac{1}{3}
\cdot
\Pr{s_{i+1}}\sigma{W}
,
\)
given the initial conditions for $i=0,1$,
yields
\(
\Pr{s_i}\sigma{W}=2^{-i}.
\)
\end{proof}

\begin{lemma}
\label{lem:val}
$\val{s_1}=\frac{3}{4}$.
\end{lemma}
\begin{proof}
First we prove
$\val{s_1}\geq\frac{3}{4}$.
For any $n$ consider the memoryless strategy, $\sigma_n$,
given by
$\sigma_n(s_i)(u_i)=1$ if $i<n$ and 
$\sigma_n(s_i)(r_i)=1$ if $i\geq n$.
Set
\(
p_i \coloneqq \Pr{s_1}{\sigma_i}{\text{Reach $s_i$}}.
\)
Observe that
$p_i$ does not change if we define it using any $\sigma_n$
with $n\geq i$, and that
\(
1-p_i = \Pr{s_1}{\sigma_n}{W\land\neg\text{Reach $s_i$}}
\)
for $n\geq i$.
Moreover, $p_1=1$ and
\(
p_{i+1}
\coloneqq
\frac{2}{3} \cdot \left( p_i + (1-p_i)\cdot p_{i+1}\right).
\)
This uniquely determines that
\(
p_i = \frac{2^{i-1}}{2^i-1}.
\)
Finally, observe that
\(
\Pr{s_1}{\sigma_n}{W}
=
(1-p_n) + p_n\cdot\frac{1}{2}
,
\)
thus
\(
\val{s_1}
\geq
\lim_{n\to\infty}
(1-p_n) + p_n\cdot\frac{1}{2}
=
\frac{3}{4}
.
\)

Now we prove that
$\val{s_1}\leq\frac{3}{4}$
by proving
\(
\Pr{s_1}{\sigma}{W}
\leq
\frac{3}{4}
\)
for all $\sigma$.
Consider the following probabilities:
\(
p_a \coloneqq
\Pr{s_1}{\sigma}{W \land \neg\text{Reach some $r_j$}}
\),
\(
p_b \coloneqq
\Pr{s_1}{\sigma}{W \land \text{Reach some $r_j$}}
\),
\(
p_c \coloneqq
\Pr{s_1}{\sigma}{\text{Reach some $r_j$}}
\).
Clearly $p_b = \frac{p_c}{2}$.
Due to Lemma~\ref{lem:prob-oc} applied to $i=1$ we also have that $p_a \leq \frac{1}{2}$.
Finally, $p_a+p_c\leq 1$ since the events are disjoint.
We conclude that
\(
\Pr{s_1}{\sigma}{W}
=
p_a+p_b
\leq
p_a + \frac{1}{2}\cdot(1-p_a)
=
\frac{1}{2}\cdot p_a
+
\frac{1}{2}
\leq
\frac{3}{4}
.
\)
\end{proof}

\begin{lemma}
\(
\val{s_i} = \frac{2^i+1}{2^{i+1}}
\)
for all $i\geq0$.
\end{lemma}
\begin{proof}
The case $i=0$ is trivial, and $i=1$ is Lemma~\ref{lem:val}.
Solving the recurrence
\(
\val{s_i}
=
\frac{2}{3}
\cdot
\val{s_{i-1}}
+
\frac{1}{3}
\cdot
\val{s_{i+1}}
,
\)
given the initial conditions for $i=0,1$,
yields
\(
\val{s_i} = \frac{2^i+1}{2^{i+1}}
.
\)
\end{proof}

In particular, for all $i\geq 1$, $\val{s_i} > \val{r_i}$,
thus no optimal strategy may use transitions $s_n\tran{} r_n,\ n\geq1$.
By Lemma~\ref{lem:prob-oc}, there are no optimal strategies in $s_i$.

\bibliographystyle{eptcs}
\bibliography{bibliography}
\end{document}